\documentclass[graybox]{svmult}

\usepackage{mathptmx}       
\usepackage{helvet}         
\usepackage{courier}        
\usepackage{type1cm}        
\usepackage{makeidx}         
\usepackage{graphicx}        
\usepackage{multicol}        
\usepackage[bottom]{footmisc}

\usepackage{amsmath,amsgen,latexsym}
\usepackage{amstext,amssymb,amsfonts}
\usepackage{comment}

\usepackage{url}

\makeindex             

\begin{document}

\title*{A Distributed Algorithm for Spectral Sparsification of Graphs with Applications to Data Clustering}
\titlerunning{A Distributed Algorithm for Spectral Sparsification of Graphs}
\author{Fabricio Mendoza-Granada and Marcos Villagra}
\institute{N\'ucleo de Investigaci\'on y Desarrollo Tecnol\'ogico (NIDTEC) \at Facultad Polit\'ecnica - Universidad Nacional de Asunci\'on, San Lorenzo C.P. 2169, Paraguay.
}

%
%
\maketitle

\abstract{Spectral sparsification is a technique that is used to reduce the number of non-zero entries in a positive semidefinite matrix with little changes to its spectrum. In particular, the main application of spectral sparsification is to construct sparse graphs whose spectra are close to a given dense graph. We study spectral sparsification under the assumption that the edges of a graph are allocated among sites which can communicate among each other. In this work we show that if a graph is allocated among several sites, the union of the  spectral sparsifiers of each induced subgraph give us an spectral sparsifier of the original graph. In contrast to other works in the literature, we present precise computations of the approximation factor of the union of spectral sparsifiers and give an explicit calculation of the edge weights. Then we present an application of this result to data clustering in the Number-On-Forehead model of multiparty communication complexity when input data is allocated as a sunflower among sites in the party.}


\section{Introduction}
\label{sec:1}
\textit{Spectral sparsification} is a technique introduced by Spielman and Teng \cite{spielman:2011} that is used to approximate a graph $G$ by a sparse graph $H$. The notion of approximation used by spectral sparsification is that the spectra of both $H$ and $G$ must be close up to a constant factor. Batson, Spielman and Srivastava \cite{batson:12} proved that every graph $G$ has an spectral sparsifier with a number of edges linear  in the number of vertices of $G$ and provided an algorithm achieving such bound. There are several algorithms in the literature that construct spectral sparsifiers of graphs with a trade-off between running time and number of edges of $H$. To the best of our knowledge, Lee and Sun \cite{lee:2018} has the best probabilistic algorithm for spectral sparsification with a running time that is almost linear and constructs spectral sparsifiers with $O(qn/\epsilon^2)$ edges, where $n$ is the number of vertices of $G$, $\epsilon$ is an approximation factor and $q\geq 10$ is a constant.

There are situations where algorithms need to work with data that is not centralized and allocated in different sites. One way to deal with decentralized data is to design communication protocols so that the sites can communicate among them. The efficiency of a communication protocol can be measured by the number of bits shared among the sites and such a measure is known as the \textit{communication complexity} of the protocol \cite{nisan:93}. When data comes in the form of a graph, the edges greatly affects communication complexity, and hence, computing spectral sparsifiers of graphs in distributed systems is of great importance.

In this work we present a distributed algorithm for spectral sparsification of graphs in the communication complexity model. In this model, we are only interested in the communication costs among sites and we assume that each site has arbitrary computational power. The idea behind this protocol is that, given an input graph $G$, spectral sparsifiers of induced subgraphs of $G$ can be computed in each site first, and then any given site computes the union of such graphs which results in a spectral sparsifier of $G$. Even though other works have used the idea of taking the union of spectral sparsifiers like Chen \emph{et al.} \cite{chen:16}, they have not shown a precise calculation of the approximation factor. The main contribution of this work, presented in Theorem \ref{th:uniones}, is an estimation of the approximation factor and an explicit calculation of the edge weights in the union of spectral sparsifiers. In order to compute the approximation factor we introduce an idea that we call ``overlapping cardinality partition,'' which is a way to partition the edge set of a graph with respect to the number of times each edge is allocated among sites. Overlapping cardinality partition is a technical tool that allows us to express the Laplacian matrix of the union of induced subgraphs of $G$ as a linear combination of the Laplacian matrices of graphs induced from the partition.

In a second part of this paper, we present in Section \ref{sec:clustering} an application of Theorem \ref{th:uniones} in distributed data clustering in the Number-On-Forehead model of communication complexity. In particular, if we assume the existence of a sunflower structure \cite{erdos:61,deza:74,kostochka:00} on the input data, we show how a communication protocol can detect the presence of the sunflower and take advantage of its kernel to reduce the communication costs. 

The rest of this paper is organized as follows. In Section \ref{sec:preliminaries} we present the main definitions and notation used throughout this work. In Section \ref{sec:dist} we present the main result of this work, and in Section \ref{sec:clustering} we present our application to data clustering.

\section{Preliminaries and Notation}\label{sec:preliminaries}
In this section we will introduce some definitions and notations that will be used throughout this paper.

\subsection{Spectral Graph Theory}
Let $G=(V,E,w)$ be an undirected and weighted graph with $n$ vertices and $m$ edges. Let $\{E_i\}_{i\geq 1}$ be family of subsets of $E$. We denote by $G_i=(V,E_i,w_i)$ the subgraph induced by  $E_i$, where $w_i:E_i\rightarrow \mathbb{R}^+$ is defined as $w_i(e)=w(e)$ for all $e\in E_i$ and 0 otherwise. Every graph $G$ has an associated matrix called its \textit{Laplacian} matrix, or simply Laplacian, which is define as
\begin{align*}
L_G=D_G-W_G,
\end{align*}
where $W_G$ is the weighted adjacency matrix and $D_G$ is the weighted degree matrix. We will omit the subindex $G$ from $L_G,W_G$ and $D_G$ when it is clear from the context.

The normalized Laplacian is defined as $\mathcal{L}=D^{-1/2}LD^{-1/2}$. The Laplacian matrix (and normalized Laplacian) is positive semidefinite (PSD) with its first eigenvalue $\lambda_1$ always equals zero with multiplicity equal to the number of connected components of $G$ \cite{luxburg:07}. Indeed, if there exists a multicut of size $k$ in $G$ then the $k$-th smallest eigenvalue $\lambda_k$ of $L$ gives useful information to find a multicut.

One of the fastest methods to approximate an optimal multicut in a graph is the so-called \textit{spectral clustering} algorithm. This technique uses $k$ eigenvectors of $L$ or $\mathcal{L}$ associated to the first $k$ smallest eigenvalues in order to construct a matrix $X$ with the eigenvectors as columns, and then, it applies a simpler clustering algorithm (like \textit{k-means}) to the rows of $X$ \cite{luxburg:07}.  Lee, Gharan, and Trevisan \cite{lee2:2014} proved that $\lambda_k$ approximates the optimal value of a multicut of size $k$ in $G$ and the eingevectors give the corresponding partition over $V$.

\subsection{Spectral Sparsification}
Spectral sparsification is a technique used to reduce the density of a given PSD matrix changing its spectra only by a constant factor of approximation. Given a matrix $M$, spectral sparsification constructs another matrix which is ``similar'' to $M$ in some well-defined way. We will use a notion of similarity defined in \cite{spielman:2011}. A subgraph $H$ of $G$ is called an $\epsilon$-spectral sparsifier of $G$ if for any $x\in \mathbb{R}^n$ we have that
\begin{align*}
(1-\epsilon)x^TL_Gx \leq x^TL_Hx \leq (1+\epsilon)x^TL_Gx.
\end{align*}
The importance of a spectral sparsifier lies on the sparseness of $L_H$, for example, some computations are easier over an sparse matrix. There are deterministic and probabilistic algorithms to find  spectral sparsifiers of a given graph. The algorithm of Batson, Spielman and Srivastava \cite{batson:12} is currently the best deterministic algorithm. The algorithm of \cite{batson:12} constructs a graph with $O(\frac{qn}{\epsilon^2})$ edges in $O(\frac{qmn^{5/q}}{\epsilon^{4+4/q}})$ time, where $\epsilon$ is the approximation factor and $q\geq 10$ is a constant.

\section{A Distributed Algorithm for Spectral Sparsification}\label{sec:dist}
In this section we present our main result. In particular, given a graph $G$ and a family of induced subgraphs of $G$, we show that the union of spectral sparsifiers of the induced subgraphs is a spectral sparsifier of $G$. In contrast to other work, however, we give explicit bounds on the approximation factor and a construction of the new weight function. 

First we introduce some definitions which will help us understand the overlapping of data among the sites. We denote by $[n]$ the set $\{1,2,\dots,n\}$.

\begin{definition}[Occurrence Number]\label{def:occurrence}
Let ${\mathcal{E}}=\{E_1,\dots, E_t\}$ be a family of subsets of $[n]$. For any $a\in [n]$, the occurrence number of $a$ in $\mathcal E$, denoted $\#(a)$, is the maximum number of sets from $\mathcal E$ in which $a$ appears.
\end{definition}
\begin{example}\label{ex:occurrence}
Let $n=7$ and ${\mathcal{E}}=\{\{1,2,3\}, \{2,3,4\}, \{4,5,1\}, \{3,2,6\},\{4,7,1\}$, $\{2,3\}$, $\{5,6,7\}, \{1,3,5\}, \{2,4\}\}$. Here we have that $\#(1)=4$, $\#(2)=5$, $\#(3)=5$, and so on.\qed
\end{example}

\begin{definition}[Overlapping Cardinality]\label{def:overcar}
Let $\mathcal{E}=\{E_1,\dots, E_t\}$ be a family of subsets of $[n]$ for some fixed $n$ and $E=\bigcup_{i=1}^tE_i$. The \emph{overlapping cardinality} of a subset $E'\subseteq E$ in $\mathcal{E}$ is a positive integer $k$ such that for each $a\in E'$ its ocurrence number $\#(a)=k$; otherwise the overlapping cardinality of $E'$ in $\mathcal E$ is $0$.
\end{definition}

The overlapping cardinality identifies the maximum number of times the elements of a subset appears in a family of subsets.

\begin{example}\label{ex:overcar}
Let $n=7$ and $\mathcal{E}$ be as in Example \ref{ex:occurrence}.
Here we have that $E=\bigcup_{i=1}^t E_i=[n]$. Now consider the sets $\{1,4\}$ and $\{1,2,3\}$.
\begin{itemize}
\item The overlapping cardinality of $\{1,4\}$ in $\mathcal E$ is 4, because $\#(1)=\#(4)=4$.
\item The overlapping cardinality of $\{1,2,
3\}$ in $\mathcal E$ is $0$ because the occurrence number of one of the elements of the set is different from the others, namely, $\#(1)=4$, $\#(2)=5$, and $\#(3)=5$.\qed
\end{itemize}
\end{example}

Now we use the idea of overlapping cardinality  to construct a partition on the set $\mathcal{E}$ of subsets of $[n]$.

\begin{definition}[Overlapping Cardinality Partition]\label{def:overcarpar}
Given a family ${\mathcal{E}}$ as in Definition \ref{def:overcar}, an \emph{overlapping cardinality partition} over $E$ on $\mathcal E$ is a partition $\{E_1',\dots, E_k'\}$ of $E$ where each $E_i'$ has overlapping cardinality $c_i$ on $\mathcal E$. We call the sequence $(c_1,c_2,\dots,c_k)$, with $1\leq c_1 < c_2 < \cdots < c_k$, the \emph{overlapping cardinalities} over the family $\mathcal{E}$.
\end{definition}

\begin{example}
Take $\mathcal E$ from examples \ref{ex:occurrence} and \ref{ex:overcar}. An overlapping cardinality partition is
\begin{align*}
\{ \{6,7\}, \{5\}, \{1,4\}, \{2,3\} \}.
\end{align*}
Here, $\{6,7\}$ has overlapping cardinality equal to $2$ because $\#(6)=\#(7)=2$. The subset $\{5\}$ has overlapping cardinality equal to $3$ because $\#(5)=3$.
In Example \ref{ex:overcar} we saw that the subset $\{1,4\}$ has overlapping cardinality $4$. Finally, the subset $\{2,3\}$ has overlapping cardinality equal to $5$ because $\#(2)=\#(3)=5$.\qed
\end{example}

Our main technical lemma shows that the Laplacian of an input graph can be rewritten as a linear combination of Laplacians corresponding to induced subgraphs constructed from an overlapping cardinality partition of the set of edges.

For the rest of this section we make the following assumptions. Let $G=(V,E,w)$ be an undirected and weighted graph with a function $w:E\rightarrow \mathbb{R}^+$, let $\mathcal{E}=\{E_1,\dots,E_t\}$ be a collection of subsets of $E$ such that $\bigcup_{i=1}^tE_i=E$ where $E_i \neq \emptyset$ and $G_i=(V,E_i,w_i)$ is an induced subgraph of $G$ where $w_i:E_i\rightarrow \mathbb{R}^+$ and $w_i(e)=w(e)$ for all $e\in E_i$ and 0 otherwise.

\begin{lemma}\label{lem:particiones}
If $1\leq c_1 < c_2 < \dots < c_k$ are the overlapping cardinalities over the family $\mathcal{E}$ with an overlapping cardinality partition $\{E_{c_j}'\}_{j\leq k}$, then $\sum_{i=1}^tL_{G_i}=\sum_{j=1}^kc_jL_{G_{c_j}'}$ where $L_{G_{c_j}'}$ is the Laplacian of $G_{c_j}'=(V,E_{c_j}',w_{c_j}')$.
\end{lemma}
\begin{proof}

First notice that, for all $e=xy\in E_{c_j}'$ there exists a subfamily of $\mathcal{E}$ with cardinality equal to $c_j$ such that $e$ belongs to every member of it and its associated subgraph. 
Take any $xy\in E_{c_j}'$ for some $j\in\{1,\dots,k\}$. There exists $c_j$ induced subgraphs $G_{i_1},\dots,G_{i_{c_j}}$ of $G$ that has $xy$ as an edge, and all other induced subgraphs $G_{k_1},\dots,G_{k_\ell}$ do not have $xy$ as and edge, where $c_j+\ell=t$. This means that
\begin{equation}\label{eq:non-diag}
\sum_{i=1}^t L_{G_{i}}(x,y)=c_j\cdot L_{G_{c_j}'}(x,y)=-c_j\cdot w(x,y).
\end{equation}
Now, let $d_G(x)$ denote the degree of $x$ in $G$. We know that $d_G(x)=\sum_{y}w(x,y)$ where $xy\in E$. Since $\{E_{c_j}'\}_{j\leq k}$ is a partition of $E$, we can rewrite the degree of $x$ as
\begin{align*}
    d_G(x)=\sum_{xy_{c_1}\in E_{c_1}'}w(x,y_{c_1})+\dots + \sum_{xy_{c_k}\in E_{c_k}'}w(x,y_{c_k}).
\end{align*}
Then, the degree of $x$ in the graph $G_{c_j}'$ is
\begin{align*}
    L_{G_{c_j}'}(x,x) = \sum_{xy_{c_j}\in E_{c_j}'}w(x,y_{c_j})=d_{G_{c_j}'}(x).
\end{align*}
If we take an edge $xy_{c_j}\in E_{c_j}'$, where $x$ is fixed, we know that $xy_{c_j}$ appears only in the induced subgraphs $G_{i_1},\dots,G_{i_{c_j}}$, and hence, we obtain
\begin{equation}\label{eq:diag}
    \sum_{i=1}^t \left(\sum_{xy_{c_j}\in E_{c_j}'} w_i(x,y_{c_j})\right)= c_j\cdot d_{G_{c_j}'}(x).
\end{equation}
If we take another edge $uv\in E_{c_m}'$, with $m \neq j$, note that $uv$ does not belong to any of the graphs $G_{i_1},\dots,G_{i_{c_j}}$ and each Laplacian matrix $L_{G_{i_1}},\dots,L_{G_{i_{c_j}}}$ has 0 in its $(u,v)$-entry. Therefore, adding $uv$ to Eq.(\ref{eq:non-diag}) we have that
\[
\sum_{i=1}^t \left(L_{G_{i}}(x,y)+L_{G_{i}}(u,v)\right)=  c_j\cdot L_{G_{c_j}'}(x,y) + c_m\cdot L_{G_{c_m}'}(u,v).
\]
Extending this argument to all equivalent classes in $\{E_{c_j}'\}_{j\leq k}$, for each non-diagonal entry $(x,y)$, with $xy\in E$, it holds
\begin{equation}\label{eq:non-diag-total}
\sum_{i=1}^t L_{G_{i}}(x,y)=  \sum_{j=1}^k c_j\cdot L_{G_{c_j}'}(x,y).
\end{equation}
A similar argument can be made for the diagonal entries with Eq.(\ref{eq:diag}), thus obtaining
\begin{equation}\label{eq:diag-total}
\sum_{i=1}^t\left(\sum_{xy_{c_1}\in E_{c_1}'} w_i(x,y_{c_1})+\cdots+\sum_{xy_{c_k}\in E_{c_k}'} w_i(x,y_{c_k})\right)
=\sum_{i=1}^t L_{G_i}(x,x)
=\sum_{j=1}^k c_j\cdot L_{G_{c_j}'}(x,x).
\end{equation}
Equations (\ref{eq:non-diag-total}) and (\ref{eq:diag-total}) imply the lemma.\qed
\end{proof}

Now we will use Lemma \ref{lem:particiones} to show that the spectral sparsifier of $\sum_{j=1}^k c_jL_{G_{c_j}'}$ is an spectral sparsifier of the Laplacian $L_{G}$ of an input graph $G$.

\begin{theorem}\label{th:uniones}
    Let $(1=c_1 < c_2 < \dots < c_k)$ be the overlapping cardinalities over the family $\mathcal{E}$ with $\{E_{c_j}'\}_{j\leq k}$ its associated overlapping cardinality partition and $L_{G_1},\dots, L_{G_t}$ the Laplacians of $G_1,\dots, G_t$. If $H_i=(V,D_i,h_i)$ is an $\epsilon$-spectral sparsifier of $G_i$, then $H=(V,\bigcup_{i=1}^tD_i,h)$ is an $\epsilon'$-spectral sparsifier of $G$ where $h(e)=\frac{\sum_{i=1}^th_i(e)}{c_1c_k}$ and $\epsilon' \geq 1-\frac{1-\epsilon}{c_k}$.
\end{theorem}
\begin{proof}
    Let $L_{H_i}$ be the Laplacian of $H_i$. By hypothesis we have that for every $i\in [t]$ and $x\in \mathbb{R}^{V}$
    \begin{align*}
    (1-\epsilon)x^TL_{G_i}x\leq x^TL_{H_i}x\leq (1+\epsilon)x^TL_{G_i}x.
    \end{align*}
    Then we may take the summation over all $i\in [t]$ to get
    \begin{align}\label{eq:1}
    (1-\epsilon)\sum_{i=1}^tx^TL_{G_i}x\leq \sum_{i=1}^tx^TL_{H_i}x\leq (1+\epsilon)\sum_{i=1}^tx^TL_{G_i}x.
    \end{align}

    Now, lets consider the left hand side of the Equation (\ref{eq:1}). Using Lemma \ref{lem:particiones} we get
    \begin{align}
    (1-\epsilon)\sum_{i=1}^t x^TL_{G_i}x &=\nonumber (1-\epsilon)\sum_{i=1}^k c_i\cdot x^TL_{G_{c_i}'}x\nonumber\\ 
    &\geq (1-\epsilon)c_1\sum_{i=1}^k x^TL_{G_{c_i}'}x\nonumber \\
    &= (1-\epsilon)c_1 x^TL_Gx,\label{eq:left}
    \end{align}
    where the last equality follows from the fact that $\{E_{c_j}'\}_{j\leq k}$ is a partition of $E$. Similarly for the right hand side of Equation (\ref{eq:1}) we have that
    \begin{align}
    (1+\epsilon)\sum_{i=1}^tx^TL_{G_i}x \leq (1+\epsilon)c_k x^TL_Gx.\label{eq:right}
    \end{align}

    Therefore, by multiplying equations (\ref{eq:left}) and (\ref{eq:right}) by $\frac{1}{c_1c_k}$ we obtain
    \begin{align*}
    (1-\epsilon)\frac{x^TL_Gx }{c_k}\leq x^TL_{H}x \leq (1+\epsilon)\frac{x^TL_Gx}{c_1},
    \end{align*}
    where $x^TL_{H}x=(\sum_{i=1}^t x^TL_{H_i}x)/(c_1c_k)$.

    To finish the proof, note that we want $1-\epsilon'\leq (1-\epsilon)/c_k$ and $(1+\epsilon)/c_1 \leq 1+\epsilon'$ with $\epsilon\leq \epsilon'<1$. In order to solve this, we choose an $\epsilon' \geq 1-\frac{1-\epsilon}{c_k}$. First notice that $1-\epsilon' \leq 1-1+\frac{1-\epsilon}{c_k}=\frac{1-\epsilon}{c_k}$. Then we have that $\frac{1+\epsilon}{c_1}\leq \frac{1+\epsilon'}{c_1}= 1+\epsilon'$. \qed
\end{proof}

From Theorem \ref{th:uniones}, a distributed algorithm for computing spectral sparsifiers is natural. Just let  every site compute a spectral sparsifier of its own input and then each site sends its result to a coordinator that will construct the union of all spectral sparsifiers.

\section{Data Clustering in the Number-On-Forehead Model}\label{sec:clustering}
In this section we will show an application of Theorem \ref{th:uniones} to distributed data clustering in the Number-On-Forehead model of communication complexity for the case when the input data is allocated as a sunflower among sites.

Clustering is an unsupervised machine learning task that involves finding a partition over a given set of points $x_1,\dots, x_n\in \mathbb{R}^d$. Such a partition must fulfill two conditions, (i) every two points in the same set must be ``similar'' in some way and (ii) every two points on different sets must be far from being similar. Each equivalence class from the partition is also called a \emph{cluster}. Clustering can be accomplished by different kinds of techniques, where \textit{spectral clustering} \cite{luxburg:07} is one of the fastest methods.

It is easy to see clustering as a graph problem, where each point corresponds to a vertex in a complete graph and the cost of each edge is interpreted as a similarity between points. Thus, finding a set of optimal clusters in data is equivalent to finding an optimal multicut in a graph.  Since the optimal multicut depends on the spectrum of the graph's Laplacian \cite{lee2:2014} and we want to keep the communication costs low, each site must be capable of constructing sparse induced subgraphs of its own data while preserving the spectrum of its graph Laplacian. 

In our communication protocol, each site is assigned an induced subgraph of $G$, and we want each site to be aware of all clusters in the data. Consequently, each site must be capable of running a clustering algorithm on its own data, communicate its results to the other sites, and then use the exchanged messages to construct an approximation to the original graph $G$. This is where the distributed spectral sparsification algorithm is relevant.

First, we will construct a protocol to verify if the input data in every site is a sunflower. If the input is indeed allocated in a sunflower structure, then a party can take advantage of the sunflower to find an approximation of clusters in the data.

\subsection{Models of Communication and their Complexity Measure}
We will introduce some standard notations from communication complexity---we refer the interested reader to the textbook by Kushilevitz and Nisan \cite{kushilevitz:97} for more details. Let $P_1,P_2,\dots,P_s$ be a set of sites where a site $P_j$ has an input $x_j\in \{0,1\}^r$, with $r$ a positive integer. In a multiparty communication protocol, with $s\geq 3$, the sites want to jointly compute a function $f:\{0,1\}^r\times \dots \times \{0,1\}^r \rightarrow Z$ for some finite codomain $Z$. In the \emph{Number-On-Forehead} model of communication, or NOF model, each site only has access to the other sites's input but not its own, that is, a site $P_j$ has access to $(x_1,...,x_{j-1},x_{j+1},...,x_s)$. In order to compute $f$ the sites must communicate, and they do so by writing bits on a blackboard which can be accessed by all sites in the party. This is the so-called \emph{blackboard model} of communication.

The maximum number of bits exchanged in the protocol over the worst-case input  is the \textit{cost} of the protocol. The \textit{deterministic communication complexity} of the function $f$ is the minimum cost over all protocols which compute $f$.

Let $G=(V,E)$ be an input graph and $\{E_j\}_{j\leq s}$ be a family of subsets of $E$. In order to study communication protocols for graph problems we assume that $E_j$ is the input data to site $P_j$. In the NOF model, we let $F_j=\{E_1,E_2,...,E_{j-1},E_{j+1},...,E_s\}$ be the set of edges which $P_j$ can access. Given a site $P_j$, the \emph{symmetric difference on $P_j$}, denoted $\Delta_j$, is defined as the symmetric difference among all sets $P_j$ has access to, that is, $\Delta_j$ is the symmetric difference between each set in $F_j$.

For the rest of this paper, we use as a shorthand $\mathcal{E}$ for the set $\{E_1,\dots, E_s\}$ of subsets of the set of edges $E$ of an input graph $G=(V,E)$ with $\bigcup_{i=1}^s E_i=E$, and $\mathcal{F}$ for the set $\{F_1,\dots, F_s\}$ where $F_j=\{E_1,\dots,E_{j-1},E_{j+1},\dots,E_s\}$. Here $\mathcal{F}$ captures the idea of the NOF model where every site have access to the other's sites input but not its own.

\subsection{Sunflowers and NOF Communication}\label{sec:sunflowers}
A \emph{sunflower} or \emph{$\Delta$-System} is a family of sets $\mathcal{A}=\{A_1,...,A_t\}$ where $(A_i\cap A_j)=\bigcap_{k=1}^t A_k=K$ for all $i\neq j$. We call $K$ the kernel of $\mathcal{A}$. The family $\mathcal{A}$ is a \emph{weak $\Delta$-System} if $|A_i\cap A_j|=\lambda$ for all $i\neq j$ for some constant $\lambda$ \cite{kostochka:00}. It is known that if $\mathcal{A}$ is a weak $\Delta$-System and $|\mathcal{A}| \geq \ell^2-\ell+2$, where $\ell=\max_{i=1}^t \{A_i\}$, then $\mathcal{A}$ is a $\Delta$-System \cite{deza:74}.

We start with a simple fact that ensures the existence of $\Delta$-Systems with the same kernel in the NOF model if input data in a communication protocol is allocated as a sunflower among sites.

\begin{lemma}\label{lem:delta_nof}
    If $s=|\mathcal{E}|\geq 3$ and $\mathcal{E}$ is a $\Delta$-System with kernel $K$, then any $F_i$ is a $\Delta$-System with kernel $K$.
\end{lemma}

The following lemma states a sufficient condition for the existence of a $\Delta$-System in the input data in the NOF model with the requirement, however, that we need at least four or more sites

\begin{lemma}\label{lem:deltasystem-4}
    Let $s=|\mathcal{E}|\geq 4$. If, for all $i\in [s]$, we have that $F_i$ is a $\Delta$-System, then $\mathcal{E}$ is a $\Delta$-System.
\end{lemma}

\begin{proof}
    Suppose that $\mathcal E$ is no a $\Delta$-System, and we want to prove that for some $1\leq i\leq s$, $F_i$ is not a $\Delta$-System.
    
    With no loss of generality, suppose that there exists exactly two sets $E_i$ and $E_j$ that certify that $\mathcal E$ is not a $\Delta$-System; that is, there exists $E_i$ and $E_j$ such that $E_i\cap E_j=K'$, and, for any $a\neq i$ and $b\neq j$,  it holds that $E_a\cap E_j=E_b\cap E_i=K$, with $K\neq K'$. Now take any $F_c$, with $c$ different from $i$ and $j$. Then $F_c$ cannot be a $\Delta$-System because $E_i$ and $E_j$ belong to $F_c$ and there is at least another set in $F_c$ because $|\mathcal E|\geq 4$.\qed
    \vspace{0.2cm}
\end{proof}




Lemma \ref{lem:deltasystem-4} implies that we only need to know if all sites in a communication protocol have access to a $\Delta$-System to ensure that an entire family of input sets is a $\Delta$-System, provided there are at least 4 sites.

\begin{proposition}\label{pro:delta_sys}
    There exists a protocol that verifies if $\mathcal{E}$, with $|\mathcal E|\geq 4$, is a $\Delta$-System or not with $s-1$ bits of communication exchanged.
\end{proposition}

With Proposition \ref{pro:delta_sys}, a multiparty communication protocol with a number of sites $s\geq 4$ can check for the existence of a sunflower structure in its input data. Furthermore, if input data is allocated among sites as a sunflower, then, by Lemma \ref{lem:delta_nof}, any site immediately knows the kernel of the sunflower.

\subsection{Data Clustering with Sunflowers}
In this section, we present a NOF communication protocol that exploits the sunflower structure in input data. First, we start by defining an overlapping coefficient of the edges of $G$ which can be seen as a measure of how well spread out are the edges among sites.
\begin{definition} The overlapping coefficient on site $P_j$ is defined as $\delta(j)=\frac{|\bigcap_{i\neq j}E_i|}{|\bigcup_{i\neq j}E_i|}$ and the greatest overlapping coefficient is defined as $\delta=\max_{j\in [s]}\delta(j).$
\label{overlapping}
\end{definition}

The following proposition presents a simple protocol that makes every site aware of the entire input graph.

\begin{proposition}\label{th:protocol1}
    Let $P_j$ be a site and let ${\mathcal{E}}$ be a weak $\Delta$-System with each $|E_k|=\ell$ for $k=1,2,\dots,s$,  with a kernel of size $\lambda$. Suppose that $s\geq \ell^2-\ell+3$. If site $P_j$ sends all the edges in $\Delta_j$, then every other site will know the entire graph $G$. The number of edges this  communication protocol sends is at most $|\bigcup_{i\neq j}E_i|(1-\delta)+\ell$.
\label{sfteo}
\end{proposition}
\begin{proof}
    We will prove this proposition by showing how each site constructs the graph $G$. First, a given site $P_j$ computes $\Delta_j$ and writes it on the blackboard. Since $s\geq \ell^2-\ell+3$, by the result of Deza \cite{deza:74}, we known that ${\mathcal{E}}$ is a sunflower with kernel $K$ and by Lemma \ref{lem:delta_nof} this kernel is the same in all sites. At this point all sites $i\neq j$ know $\Delta_j$, therefore, they can construct $G$ by its own using the kernel $K$ of $\mathcal{E}$. In one more round, one of the sites $i\neq j$ writes $E_j$ so that site $P_j$ can also construct $G$.
    
    In order to compute the communication cost of the protocol, first notice that $\delta ={\lambda}/({|\bigcup_{i\neq j}E_i|})={\lambda}/({|\Delta_j|+\lambda})$,where we used the fact that the union of all edges in every site equals the union of the symmetric difference and the kernel $K$. Then we have that $\delta |\Delta_j|=\lambda-\delta \lambda$, which implies $|\Delta_j|=\frac{\lambda-\delta \lambda}{\delta}=|\bigcup_{i\neq j}E_i||(1-\delta)$, where the last equality follows from the fact that $|\bigcup_{i\neq j}E_i|=\lambda/\delta$. Finally, after $E_j$ was sent to the blackboard the communication cost is $|\bigcup_{i\neq j}E_i||(1-\delta)+\ell$.\qed
\end{proof}

\begin{theorem}\label{protocol1_spectra}
    Let ${\mathcal{E}}$ be a weak $\Delta$-system with each $|E_k|=\ell$ for $k=1,2,\dots,s$, and suppose that $s\geq \ell^2-\ell+3$. There exists a communication protocol such that after two rounds of communication every site knows an $\epsilon$-spectral sparsifier of the entire graph $G$ with communication cost $O\left(\log \left( \frac{n}{\epsilon^2}\sqrt{1-\delta}\right)\right)$.
\end{theorem}
\begin{proof}
    From \cite{deza:74} we know that $\mathcal{E}$ is a sunflower with a kernel $K$ of size $\lambda$ and, by Lemma \ref{lem:delta_nof}, $K$ is equal in all sites. First, a site $P_j$ computes a spectral sparsifier $H_j=(V,\hat{\Delta}_j)$ of the induced subgraph $G_j=(V,\Delta_j)$ using the spectral sparsification algorithm of \cite{lee:2018}. This way we have that $|\hat{\Delta}_j|=O(n/\epsilon^2)$ where $0<\epsilon\leq 1/120$. Then site $P_j$ writes $\hat{\Delta}_j$ on the blackboard. Any other site $i\neq j$ constructs an $\epsilon$-spectral sparsifier $H_i'=(V,\hat{E}_j)$ of $G_i'=(V,E_j)$. By Theorem \ref{th:uniones}, the graph $H=(V,\hat{\Delta}_j\cup \hat{E}_j)$ is a $\epsilon'$-spectral sparsifier of $G$. In a second round, a given site $P_i$ writes $\hat{E}_j$ on the blackboard. Finally, site $P_j$ receives $\hat{E}_j$ and by Theorem \ref{th:uniones} it can also construct an $\epsilon'$-spectral sparsifier for $G$. Finally, the communication complexity is upper-bounded by $O\left(\log\left(\frac{n}{\epsilon^2}(1-\delta)\right)+\log \left(\frac{n}{\epsilon^2}\right)\right) = O\left(\log \left( \frac{n}{\epsilon^2}\sqrt{1-\delta}\right)\right).$\qed
\end{proof}
\vspace{0.2cm}

\noindent\textbf{Acknowledgement.} We give our thanks to the reviewers of CTW 2020 for their comments that helped improve this paper. This work is supported by Conacyt research grants POSG17-62 and PINV15-208.

\bibliographystyle{spmpsci}
\bibliography{exampleST_ctw2020}

\begin{thebibliography}{10}
\providecommand{\url}[1]{{#1}}
\providecommand{\urlprefix}{URL }
\expandafter\ifx\csname urlstyle\endcsname\relax
  \providecommand{\doi}[1]{DOI~\discretionary{}{}{}#1}\else
  \providecommand{\doi}{DOI~\discretionary{}{}{}\begingroup
  \urlstyle{rm}\Url}\fi

\bibitem{batson:12}
Batson, J.D., Spielman, D.A., Srivastava, N.: Twice-ramanujan sparsifiers.
\newblock In: Proceedings of the 41st annual ACM symposium on Theory of
  computing (STOC), pp. 255--262 (2009)

\bibitem{chen:16}
Chen, J., Sun, H., Woodruff, D., Zhang, Q.: Communication-optimal distributed
  clustering.
\newblock In: Proceedings of the 30th International Conference on Neural
  Information Processing Systems (NIPS), pp. 3727--3735 (2016)

\bibitem{deza:74}
Deza, M.: Solution d'un probl{\`{e}}me de {E}rd\"{o}s-{L}ov\'{a}sz.
\newblock Journal of Combinatorial Theory, Series B \textbf{16}(2), 166--167
  (1974)

\bibitem{erdos:61}
Erd\"{o}s, P., Chao, Rado, R.: Intersection theorems for systems op finite
  sets.
\newblock Quarterly Journal of Mathematics \textbf{12}(1), 313--320 (1961)

\bibitem{kostochka:00}
Kostochka, A.: Extremal problems on {$\Delta$}-systems.
\newblock Numbers, Information and Complexity pp. 143--150 (2000)

\bibitem{kushilevitz:97}
Kushilevitz, E., Nisan, N.: Communication complexity.
\newblock Cambridge University Press (2006)

\bibitem{lee2:2014}
Lee, J.R., Gharan, S.O., Trevisan, L.: Multiway spectral partitioning and
  higher-order cheeger inequalities.
\newblock Journal of the ACM (JACM) \textbf{61}(6), 37 (2014)

\bibitem{lee:2018}
Lee, Y.T., Sun, H.: Constructing linear-sized spectral sparsification in
  almost-linear time.
\newblock SIAM Journal on Computing \textbf{47}(6), 2315--2336 (2018)

\bibitem{luxburg:07}
von Luxburg, U.: A tutorial on spectral clustering.
\newblock Statistics and computing \textbf{17}(4), 395--416 (2007)

\bibitem{nisan:93}
Nisan, N.: The communication complexity of threshold gates.
\newblock Combinatorics, Paul Erd\"os is Eighty \textbf{1}, 301--315 (1993)

\bibitem{spielman:2011}
Spielman, D.A., Teng, S.H.: Spectral sparsification of graphs.
\newblock SIAM Journal On Computing \textbf{40}(4), 981--1025 (2011)

\end{thebibliography}

\end{document}